\documentclass[sn-mathphys]{sn-jnl}% Math and Physical Sciences Reference Style
\usepackage{caption}
\usepackage{subcaption}
\usepackage{xcolor}
\usepackage{soul,color}
\jyear{2021}%

\theoremstyle{thmstyleone}%
\newtheorem{theorem}{Theorem}%  meant for continuous numbers
\newtheorem{proposition}{Proposition}% to get separate numbers for theorem and proposition etc.
\newtheorem{assumption}{Assumption}
\newtheorem{lemma}{Lemma}

\theoremstyle{thmstyletwo}%
\newtheorem{remark}{Remark}%

\theoremstyle{thmstylethree}%
\newtheorem{definition}{Definition}%

\begin{document}

\title[ ]{Lyapunov-based Nonlinear Model Predictive Control for Attitude Trajectory Tracking of Unmanned Aerial Vehicles}

% https://www.springer.com/journal/40435/


\author[1]{\fnm{Duy Nam} \sur{Bui}}\email{duynam.robotics@gmail.com}

\author[1]{\fnm{Thi Thanh Van} \sur{Nguyen}}\email{vanntt@vnu.edu.vn}

\author*[2]{\fnm{Manh Duong} \sur{Phung}}\email{duong.phung@fulbright.edu.vn}

% \affil*[1]{\orgdiv{Department}, \orgname{Organization}, \orgaddress{\street{Street}, \city{City}, \postcode{100190}, \state{State}, \country{Country}}}

\affil[1]{\orgname{Faculty of Electronics and Telecommunications, University of Engineering and Technology, Vietnam National University}, \orgaddress{\city{Hanoi}, \country{Vietnam}}}

\affil[2]{\orgdiv{Undergraduate Faculty}, \orgname{Fulbright University Vietnam}, \orgaddress{\city{Ho Chi Minh City}, \country{Vietnam}}}


\abstract{This paper presents a new Lyapunov-based nonlinear model predictive controller (LNMPC) for the attitude control problem of unmanned aerial vehicles (UAVs), which is essential for their functioning operation. The controller is designed based on a quadratic cost function integrating UAV dynamics and system constraints. An additional contraction constraint is then introduced to ensure closed-loop system stability. That constraint is fulfilled via a Lyapunov function derived from a sliding mode controller (SMC). The feasibility and stability of the LNMPC are finally proved. Simulation and comparison results show that the proposed controller guarantees the system stability and outperforms other state-of-the-art nonlinear controllers such as the backstepping controller (BSC) and SMC. In addition, the proposed controller can be integrated into an existing UAV model in the Gazebo simulator to perform software-in-the-loop tests. The results show that the LNMPC is better than the built-in PID controller of the UAV, which confirms the validity and applicability of our proposed approach.
}

\keywords{Unmanned aerial vehicles, attitude control, trajectory tracking, model predictive control}



\maketitle

\section{Introduction}
In recent years, unmanned aerial vehicles (UAVs) have been receiving significant interest due to their applicability in various fields from transportation and agriculture to military and space exploration \cite{8782102, LIU2020105671,PHUNG201725}. The key to the success of UAV applications lies in control techniques that drive the UAV to reach its desired trajectory and maintain its stability in conditions subject to disturbances such as load variation or wind gut. For quadrotor UAVs, it is essential to control their attitude since it decides their maneuver, i.e., the UAV changes its position and orientation by changing its attitude. Since the dynamic model of quadrotor UAVs is nonlinear with six degrees of freedom but having only four independent inputs, it is an underactuated system \cite{8756125}. Besides, system constraints are often required to meet the physical limits of the UAV. Therefore, designing controllers for quadrotor UAVs is a challenging problem that requires sufficient investigation.

In the literature, linear control methods such as the proportional-integral-derivative (PID) and linear-quadratic-regulator (LQR) controllers are among the most popular methods used for UAVs \cite{NAJM20191087, MAHMOODABADI2020105598, SIRELKHATEM20226275}. In \cite{8667170}, the PID, LQR, and state feedback controllers are used to control the attitude of a UAV with performance sufficient for several tasks. Combining linear controllers with an adaptive method is another approach to improve the tracking performance \cite{8724821, s19010024}. In \cite{8724821}, a fuzzy approach is used to automatically adjust the control parameters of a PID controller. In \cite{s19010024}, the fuzzy PID method is combined with an iterative learning controller to deal with under-actuated dynamics and strong coupling characteristics of the quadrotor. While these approaches can enhance the performance of standard PID controllers, they require linearization, which affects the control performance when the UAV operates in its nonlinear regions.

To overcome that problem, nonlinear control methods are often used. The most popular nonlinear control techniques include sliding mode control (SMC) and its variations due to the capability to handle model uncertainties \cite{8287250, ALIPOUR201916}. In \cite{8287250}, an adaptive twisting SMC algorithm is developed to control the attitudes of a quadrotor with satisfactory performance. In \cite{9373410}, an SMC based on neural networks is designed to minimize the impact of external disturbances on UAV dynamics. The SMC and its variants are also used in \cite{YANG2016208, GONG2019105444, 9410438} to regulate the UAV attitude under different conditions. However, the inevitable chattering phenomenon generated by SMC is undesirable to the system and the computational cost is proportional to the order of the system. Hence, another nonlinear control method named the backstepping control (BSC) is often used \cite{das2009backstepping, 8494719, glida2020optimal}. It can handle nonlinearities in the system dynamics and external disturbances to generate fast and efficient tracking performance. For instance, the adaptive BSC introduced in \cite{FU2018593} is capable of controlling the attitude of a quadrotor in harsh conditions with desirable accuracy. However, if the system model is subject to constraints or uncertainties, the control performance is quickly degraded and may lead to divergence. The BSC and SMC can be combined to take advantage of both controllers \cite{TUAN2019297, 7448915}, but the undesirable chattering phenomenon then persists. Besides, the system stability is affected if constraints on system states and control signals exist.

Recently, model predictive control (MPC) has been used for UAV control to predict system states and handle constraints. In \cite{Yang2013}, an adaptive nonlinear MPC is introduced for path tracking of a fixed-wing UAV. The prediction horizon varies according to the path curvature to enhance the tracking result. In \cite{cavanini2021model}, MPC is used for the autopilot of a UAV in which a linear parameter-varying model is employed to describe its dynamics and account for estimation errors. In another work, a low complexity MPC algorithm is introduced for real-time trajectory tracking of the UAV with limited computation capacity \cite{BANGURA201411773}. MPC is also used in \cite{6965772, 8375687, 7272877, XU2020105686, doi:10.2514/6.2017-1512} for different UAV control tasks such as guidance, linear tracking, and multitarget-multisensor tracking. However, the stability of those controllers is hardly addressed due to the lack of a direct relationship between control signals and input variables. Since stability is an essential characteristic of control systems, analyzing it is essential for the safe operation of UAVs.

In this work, we address the attitude control problem of UAVs by proposing a nonlinear model predictive controller. First, the dynamic model of the quadcopter UAV is introduced. A cost function relating the current and desired state and having constraints on system states and control signals is then defined. A Newton-type algorithm is finally used to solve the cost function to obtain optimal control signals. Here, our contributions are threefold:

\begin{enumerate}
    \item[(i)] We propose a new Lyapunov-based nonlinear model predictive controller (LNMPC) considering constraints on both attitude states and control signals. Those constraints are essential for the practical use of the controller since real UAVs are limited in their maneuverability and actuator power.
    \item[(ii)] We introduce a contraction condition to constrain the system states and based on it, we prove that the stability of the control system is guaranteed. To the best of our knowledge, this is the first time the stability of a non-linear MPC has been proven for attitude control of the UAV.
    \item[(iii)] A number of simulations, comparisons, and software-in-the-loop (SIL) tests have been conducted to evaluate the performance of the proposed controller. The results show that our controller outperforms not only the built-in controller of the UAV but also state-of-the-art nonlinear controllers.
\end{enumerate}

The rest of this paper is structured as follows. Section \ref{dynamics} presents the dynamic model of the quadrotor UAV. Section \ref{LNMPC} describes the design of the LNMPC. Section \ref{stability} analyzes the system stability. Finally, results and conclusions are presented in Section \ref{results} and Section \ref{SecConclusion}.

\section{The dynamic model of quadrotor UAVs}
\label{dynamics}

The UAV used in this work is a quadrotor drone that has two pairs of propellers rotating in opposite directions. By defining the inertial frame $O=\{X,Y,Z\}$ and the body frame $B=\{B_x,B_y,B_z\}$ as in Figure \ref{fig:vehicle}, a UAV configuration includes its position, $(x,y,z)$, and orientation, $\left(\phi,\theta,\psi\right)$, about x, y, and z axes of the inertial frame, respectively.
\begin{figure}
    \centering
    \includegraphics[width=0.5\textwidth]{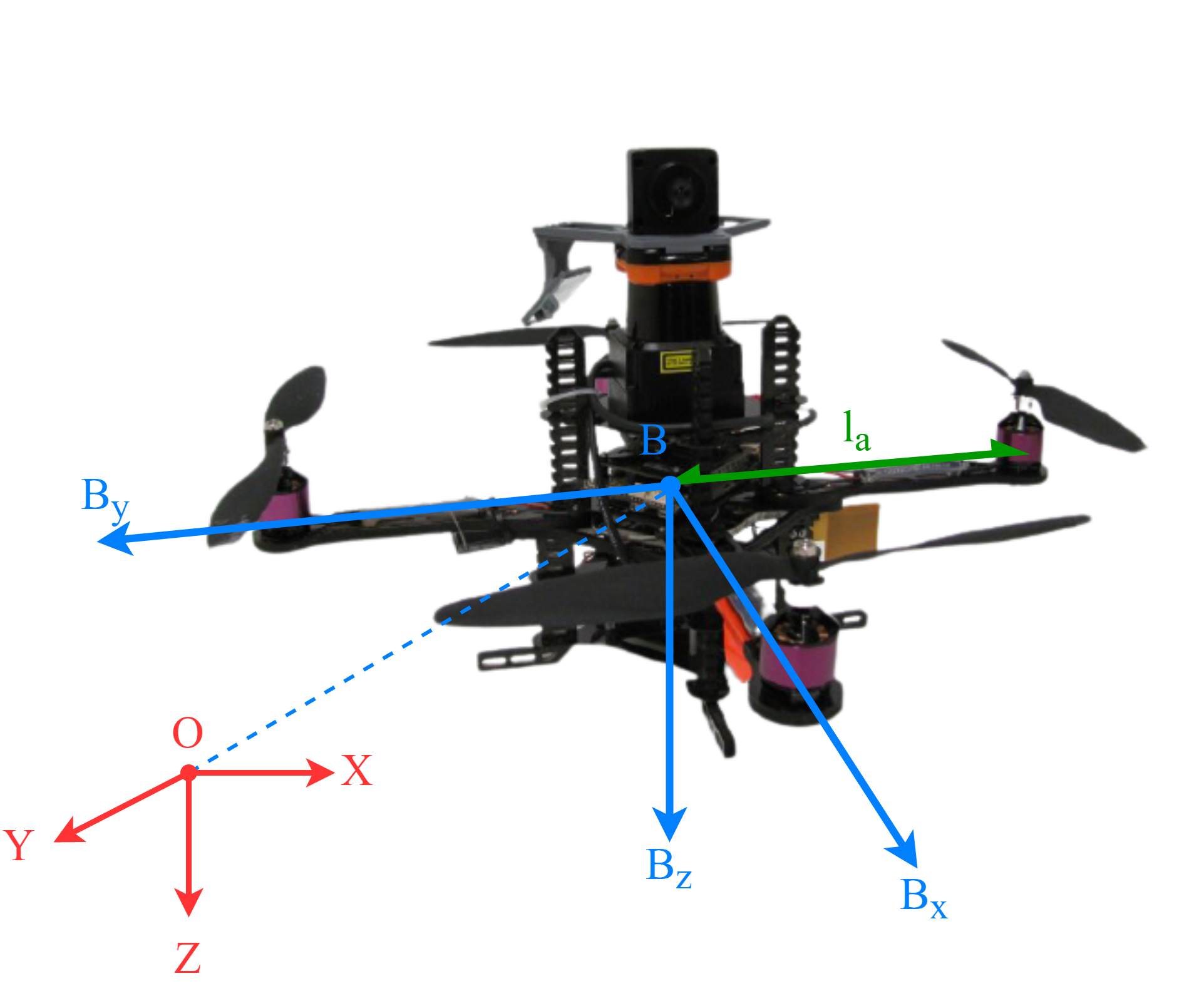}
    \caption{The UAV model and coordinate systems}
    \label{fig:vehicle}
\end{figure}

According to \cite{Wang2016}, the dynamic equations of the UAV can be summarized as
\begin{equation} 
    \left[\begin{array}{c}
    \ddot{x}\\
    \ddot{y}\\
    \ddot{z}\\
    \ddot{\phi}\\
    \ddot{\theta}\\
    \ddot{\psi}
    \end{array}\right]=\left[\begin{array}{c}
    \left(\cos\phi\sin\theta\cos\psi+\sin\phi\sin\psi\right)\dfrac{f_{t}}{m}\\[1.5ex]
    \left(\cos\phi\sin\theta\sin\psi-\sin\phi\cos\psi\right)\dfrac{f_{t}}{m}\\[1.5ex]
    g-\cos\phi\cos\theta\dfrac{f_{t}}{m}\\[1.5ex]
    \dot{\theta}\dot{\psi}\dfrac{I_{y}-I_{z}}{I_{x}}+\dfrac{l_{a}}{I_{x}}\tau_{\phi}\\[1.5ex]
    \dot{\theta}\dot{\psi}\dfrac{I_{z}-I_{x}}{I_{y}}+\dfrac{l_{a}}{I_{y}}\tau_{\theta}\\[1.5ex]
    \dot{\phi}\dot{\theta}\dfrac{I_{x}-I_{y}}{I_{z}}+\dfrac{1}{I_{z}}\tau_{\psi}
    \end{array}\right],
    \label{eq:dynamics}
\end{equation}
where $\boldsymbol{U} = \left[f_t,\tau_\phi,\tau_\theta,\tau_\psi\right]^T$ is the control input with $f_t$ being the total thrust, and $\tau_\phi$, $\tau_\theta$ and $\tau_\psi$ being the moments about x, y, and z axes, respectively; $m$ is the mass of the UAV; $g=9.8$ m/s$^2$ is the gravity; $I_x$, $I_y$, $I_z$ are respectively the moments of inertia about x, y, and z axes of the body frame; and $l_a$ is the UAV arm length.

Since this work focuses on attitude control, only three last equations of (\ref{eq:dynamics}) are considered. Therefore, the dynamic equations representing UAV's attitude are expressed as
\begin{equation}
    \left[\begin{array}{c}
\ddot{\phi}\\
\ddot{\theta}\\
\ddot{\psi}
\end{array}\right]=\left[\begin{array}{c}
\dot{\theta}\dot{\psi}\dfrac{I_{y}-I_{z}}{I_{x}}+\dfrac{l_{a}}{I_{x}}\tau_{\phi}\\[1.5ex]
\dot{\theta}\dot{\psi}\dfrac{I_{z}-I_{x}}{I_{y}}+\dfrac{l_{a}}{I_{y}}\tau_{\theta}\\[1.5ex]
\dot{\phi}\dot{\theta}\dfrac{I_{x}-I_{y}}{I_{z}}+\dfrac{1}{I_{z}}\tau_{\psi}
\end{array}\right].
\label{eqn:attitude}
\end{equation}

% \textbf{Can giai thich cho ro nghia hon}

% For the given system, both linear and nonlinear controllers can be used for reference tracking. However, constraints on states, control inputs, are not be handled in the control progress. Thus, the actuator constraints and even the stability cannot be satisfied. Therefore, in our work, all the constraints are all considered.
\section{Controller design}
\label{LNMPC}
Similar to MPC, LNMPC can handle system constraints to output optimal control signals. Moreover, the stability of LNMPC can be analyzed via the Lyapunov theory. In this section, we present our design of the LNMPC together with the design of a sliding mode controller used to prove the system's stability.

\subsection{Design of the LNMPC}
Let $\boldsymbol{\xi}_{1}=\left[\phi,\theta,\psi\right]^{T}$ and $\boldsymbol{\xi}_{2}=\left[\dot{\phi},\dot{\theta},\dot{\psi}\right]^{T}$ represent the attitude and angular velocities of the UAV, respectively. According to (\ref{eqn:attitude}), we have
\begin{equation}
\begin{aligned}
    \dot{\boldsymbol{\xi}}_{1}&=\boldsymbol{\xi}_{2},\\
    \dot{\boldsymbol{\xi}}_{2}&=\left[\begin{array}{c}
    \ddot{\phi}\\
    \ddot{\theta}\\
    \ddot{\psi}
    \end{array}\right]=\left[\begin{array}{c}
    \dot{\theta}\dot{\psi}\dfrac{I_{y}-I_{z}}{I_{x}}+\dfrac{l_{a}}{I_{x}}\tau_{\phi}\\
    \dot{\phi}\dot{\psi}\dfrac{I_{z}-I_{x}}{I_{y}}+\dfrac{l_{a}}{I_{y}}\tau_{\theta}\\
    \dot{\phi}\dot{\theta}\dfrac{I_{x}-I_{y}}{I_{z}}+\dfrac{1}{I_{z}}\tau_{\psi}
    \end{array}\right]\\
    &=\left[\begin{array}{ccc}
    \dfrac{I_{y}-I_{z}}{I_{x}} & 0 & 0\\
    0 & \dfrac{I_{z}-I_{x}}{I_{y}} & 0\\
    0 & 0 & \dfrac{I_{x}-I_{y}}{I_{z}}
    \end{array}\right]\left[\begin{array}{c}
    \dot{\theta}\dot{\psi}\\
    \dot{\phi}\dot{\psi}\\
    \dot{\phi}\dot{\theta}
    \end{array}\right]+\left[\begin{array}{ccc}
    \dfrac{l_{a}}{I_{x}} & 0 & 0\\
    0 & \dfrac{l_{a}}{I_{y}} & 0\\
    0 & 0 & \dfrac{1}{I_{z}}
    \end{array}\right]\left[\begin{array}{c}
    \tau_{\phi}\\
    \tau_{\theta}\\
    \tau_{\psi}
    \end{array}\right]\\
    &=\boldsymbol{G}_{1}g\left(\boldsymbol{\xi}_{2}\right)+\boldsymbol{G}_{2}\boldsymbol{U},
\end{aligned}
\label{eqn:at_dynamic_2}
\end{equation}
where $g\left(\boldsymbol{\xi}_{2}\right)=\left[\dot{\theta}\dot{\psi},\dot{\phi}\dot{\psi},\dot{\phi}\dot{\theta}\right]^{T}$.

Let $\boldsymbol{\xi}=\left[\boldsymbol{\xi}_{1},\boldsymbol{\xi}_{2}\right]^T$ be the state vector and  $\boldsymbol{U}=\left[\tau_{\phi},\tau_{\theta},\tau_{\psi}\right]^{T}$ be the control signal. The dynamic equations in (\ref{eqn:at_dynamic_2}) can be written in a short form as
\begin{equation}
    \dot{\boldsymbol{\xi}}=f\left(\boldsymbol{\xi},\boldsymbol{U}\right).
\label{eqn:sys}
\end{equation}

Constraints on system states and control signals can be included as 
\begin{equation}
\begin{aligned}
    &\left\vert\boldsymbol{\xi}\right\vert\leq\boldsymbol{\xi}_{max},\\
    &\left\vert\boldsymbol{U}\right\vert\leq\boldsymbol{U}_{max}.
\end{aligned}
\end{equation}

Let $\boldsymbol{\xi}_{1d}=\left[\phi_d,\theta_d,\psi_d\right]^{T}$ be the reference attitude and $\boldsymbol{\xi}_{2d}=\dot{\boldsymbol{\xi}}_{1d}$. The tracking error is given by
\begin{equation}
\boldsymbol{Z}=\left[\begin{array}{c}
    \boldsymbol{z}_{1}\\
    \boldsymbol{z}_{2}
    \end{array}\right]
=\left[\begin{array}{c}
    \boldsymbol{\xi}_{1}-\boldsymbol{\xi}_{1d}\\
    \boldsymbol{\xi}_{2}-\boldsymbol{\xi}_{2d}
    \end{array}\right].
    \label{eqn:error_vec}
\end{equation}

Denoting the reference state as $   \boldsymbol{\xi}_{d}=\left[\boldsymbol{\xi}_{1d},\boldsymbol{\xi}_{2d}\right]^T$, the aim of the LNMPC is to find optimal control signal $\boldsymbol{U}$  to drive the state variable $\boldsymbol{\xi}$ to the desired state $\boldsymbol{\xi}_{d}$. This can be carried out by optimizing a cost function with constraints of the form 
\begin{equation}
    J=\underset{\boldsymbol{U}}{\min}\left\Vert \boldsymbol{\xi}\left(t_k+T\right)-\boldsymbol{\xi}_{d}\left(t_k+T\right)\right\Vert _{\boldsymbol{P}}^{2}+\int_{t_k}^{t_k+T}\left(\left\Vert \boldsymbol{\xi}\left(t\right)-\boldsymbol{\xi}_{d}\left(t\right)\right\Vert _{\boldsymbol{Q}}^{2}+\left\Vert \boldsymbol{U}\left(t\right)\right\Vert _{\boldsymbol{R}}^{2}\right)dt
\label{eqn:lnmpc}
\end{equation}
subject to
\begin{subequations}
    \begin{equation}
        \boldsymbol{\xi}\left(t_k\right)=\boldsymbol{\xi}\left(t=t_k\right)
    \end{equation}
    \begin{equation}
        \dot{\boldsymbol{\xi}}\left(t\right)=f\left(\boldsymbol{\xi}\left(t\right),\boldsymbol{U}\left(t\right)\right)\forall \ t \in  [t_{k};t_{k}+T]
    \end{equation}
    \begin{equation}
        \left\vert\boldsymbol{\xi}\left(t\right)\right\vert\leq\boldsymbol{\xi}_{max}\forall \ t \in  [t_{k};t_{k}+T]
    \end{equation}
    \begin{equation}
        \left\vert\boldsymbol{U}\left(t\right)\right\vert\leq\boldsymbol{U}_{max}\forall \ t \in  [t_{k};t_{k}+T]
    \end{equation}
    \begin{equation}
        \dfrac{\partial V\left(\boldsymbol{\xi}\right)}{\partial\boldsymbol{\xi}}f\left(\boldsymbol{\xi}\left(t_k\right),\boldsymbol{U}\left(t_k\right)\right)\leq\dfrac{\partial V\left(\boldsymbol{\xi}\right)}{\partial\boldsymbol{\xi}}f\left(\boldsymbol{\xi}\left(t_k\right),h\left(\boldsymbol{\xi}\left(t_k\right)\right)\right),
        \label{eqn:contraction}
    \end{equation}
\label{eqn:lnmpc1}
\end{subequations}
where $t_k$ is the current time instant, $T$ is the prediction horizon length, $\boldsymbol{U}\left(t\right)$ is the prediction control signal at future time $t$, $\left\Vert\cdot\right\Vert$ denotes the Euclidean norm, $\boldsymbol{P}$, $\boldsymbol{Q}$, and $\boldsymbol{R}$ are the positive definite weight matrices, $h\left(\cdot\right)$ is a Lyapunov-based closed-loop control law, and $V\left(\cdot\right)$ is a Lyapunov function. 
% \textcolor{red}{The detailed about the proposed controller can be summarized in Figure \ref{fig:diagram}.
% }

% \begin{figure}[h!]
% \centering
%     \centering
%     \includegraphics[width=0.7\textwidth]{images/diagram.png}
%     \caption{The diagram of the LNMPC}
%     \label{fig:diagram}
% \end{figure}

The LNMPC problem (\ref{eqn:lnmpc} - \ref{eqn:lnmpc1}) can be solved by using a Newton-type algorithm called sequential quadratic programming (SQP) to find a locally optimal solution \cite{quirynen2015autogenerating}. Here, we add an additional contraction constraint (\ref{eqn:contraction}) to the original nonlinear MPC to maintain the closed-loop stability. Constraint (\ref{eqn:contraction}) guarantees that values of the first order derivative of Lyapunov function $V\left(\cdot\right)$ with control input $\boldsymbol{U}(t_k)$  are smaller than or equal to those values obtained by using the Lyapunov-based control law $h\left(\boldsymbol{\xi}(t_k)\right)$. This contraction constraint thus allows proving that the LNMPC inherits the stability property of the controller having $h\left(\boldsymbol{\xi}\right)$. According to \cite{Christofides2011}, any Lyapunov-based closed-loop controller can be used to implement $h\left(\boldsymbol{\xi}\right)$. In this work, we use a sliding mode controller for that purpose with its details described below.

\subsection{Sliding mode control design for the contraction constraint}
This section presents our design of $h\left(\boldsymbol{\xi}\right)$ in (\ref{eqn:contraction}) using the sliding mode control technique. Let us consider the sliding surface of the control system as follow:
\begin{equation}
    s=\dot{\boldsymbol{z}}_1+\lambda\boldsymbol{z}_1=\boldsymbol{z}_2+\lambda\boldsymbol{z}_1,
    \label{eqn:sliding_surface_1}
\end{equation}
where $\lambda$ is a positive gain matrix.
% Since $z_2 = \dot{z_1}$, (\ref{eqn:sliding_surface_1}) becomes
% \begin{equation}
%     s=\boldsymbol{z}_2+\lambda\boldsymbol{z}_1.
%     \label{eqn:sliding_surface_2}
% \end{equation}

Taking the first derivative of (\ref{eqn:sliding_surface_1}), we have
\begin{equation}
    \dot{s}=\dot{\boldsymbol{z}}_2+\lambda\dot{\boldsymbol{z}}_1=\left(\dot{\boldsymbol{\xi}}_{2}-\dot{\boldsymbol{\xi}}_{2d}\right)+\lambda\boldsymbol{z}_2.
\end{equation}

Consider the SMC Lyapunov candidate function
\begin{equation}
    V_\text{SMC}=\dfrac{1}{2}s^Ts.
    \label{eqn:V}
\end{equation}

Differentiating $V_\text{SMC}$ in (\ref{eqn:V}) and noting (\ref{eqn:at_dynamic_2}), we get
\begin{equation}
    \dot{V}_\text{SMC}=s^T\dot{s}=s^T\left(\boldsymbol{G}_1g\left(\boldsymbol{\xi}_2\right)+\boldsymbol{G}_2\boldsymbol{U}-\dot{\boldsymbol{\xi}}_{2d}+\lambda\boldsymbol{z}_2\right).
    \label{eqn:dotV}
\end{equation}

To ensure $\dot{V}_\text{SMC}$ negative, we choose the control signal of the form 
\begin{equation}
    \boldsymbol{U}_\text{SMC}=\boldsymbol{U}_{eq}+\boldsymbol{U}_{sw},
    \label{eqn:U}
\end{equation}
where $\boldsymbol{U}_{eq}$ is the equivalent control signal that maintains the system state on the sliding manifold and $\boldsymbol{U}_{sw}$ is the signal that leads the system to the sliding surface. They are chosen as follows:
\begin{equation}
\begin{aligned}
    &\boldsymbol{U}_{eq}=\boldsymbol{G}_{2}^{-1}\left(\dot{\boldsymbol{\xi}}_{2d}-\lambda\boldsymbol{z}_{2}-\boldsymbol{G}_{1}g\left(\boldsymbol{\xi}_{2}\right)\right)\\
    &\boldsymbol{U}_{sw}=-\boldsymbol{G}_{2}^{-1}\left(c_{1}\text{sign}\left(s\right)+c_{2}s\right),
\end{aligned}
\label{eqn:U_detail}
\end{equation}
where $c_1$ and $c_2$ are positive diagonal matrices. Substituting (\ref{eqn:U}) and (\ref{eqn:U_detail}) into (\ref{eqn:dotV}) gives
\begin{equation}
    \dot{V}_\text{SMC}=-s^{T}\left(c_{1}\text{sign}\left(s\right)+c_{2}s\right)\leq0.
    \label{eqn:negativeV}
\end{equation}

Thus, the Lyapunov stability of the SMC is satisfied. From (\ref{eqn:negativeV}), the detailed expression of the contraction constraint (\ref{eqn:contraction}) for the LNMPC corresponding to (\ref{eqn:dotV}) becomes 
\begin{equation}
\begin{aligned}
    \dot{V}_\text{LNMPC}\left(t_k\right)&=s\left(t_k\right)^{T}\left(\boldsymbol{G}_{1}g\left(\boldsymbol{\xi}_{2}\left(t_k\right)\right)+\boldsymbol{G}_{2}\boldsymbol{U}\left(t_k\right)-\dot{\boldsymbol{\xi}}_{2d}\left(t_k\right)+\lambda\boldsymbol{z}_{2}\left(t_k\right)\right)\\
    &\leq-s\left(t_k\right)^{T}\left(c_{1}\text{sign}\left(s\left(t_k\right)\right)+c_{2}s\left(t_k\right)\right).
\end{aligned}
\end{equation}

\section{Stability analysis}
\label{stability}
In this section, we analyze the recursive feasibility and closed-loop stability of the LNMPC. First, we consider the following definition and lemma:
\begin{definition}
If $x\in\mathbb{R}^{n}$, the maximum norm of $x$ is
\begin{equation}
    \left\Vert x\right\Vert _{\infty}=\max\left\{ \left\vert x_{k}\right\vert:1\leq k\leq n\right\}.
\end{equation}

For $1\leq p<\infty$, the p-norm of $x$ is 
\begin{equation}
    \left\Vert x\right\Vert _{p}=\left(\sum_{k=1}^{n}\left\vert x_{k}\right\vert^{p}\right)^{1/p}.
\end{equation}
\label{def:1}
\end{definition}

\begin{lemma}
For any square matrices $\boldsymbol{M},\boldsymbol{N}\in\mathbb{R}^{n\times n}$,
\begin{equation}
    \left\Vert \boldsymbol{M}\boldsymbol{N}\right\Vert _{\infty}\leq\left\Vert \boldsymbol{M}\right\Vert _{\infty}\left\Vert \boldsymbol{N}\right\Vert _{\infty}.
\end{equation}
\label{lem:1}
\end{lemma}
Then, we make the following assumptions.
\begin{assumption}
The desired and real configurations of the UAV are smooth and bounded, i.e.,
\begin{equation}
    \left\Vert \boldsymbol{\xi}_{1d}\left(t\right)\right\Vert _{\infty}\leq\bar{\xi}_{1d}, \left\Vert \dot{\boldsymbol{\xi}}_{1d}\left(t\right)\right\Vert _{\infty}\leq\bar{\xi}_{2d},\left\Vert \ddot{\boldsymbol{\xi}}_{1d}\left(t\right)\right\Vert _{\infty}\leq\bar{\xi}_{3d}
\end{equation}
\begin{equation}
    \left\Vert \boldsymbol{\xi}_{1}\left(t\right)\right\Vert _{\infty}\leq\bar{\xi}_{1},
    \left\Vert \dot{\boldsymbol{\xi}}_{1}\left(t\right)\right\Vert _{\infty}\leq\bar{\xi}_{2},\left\Vert \ddot{\boldsymbol{\xi}}_{1}\left(t\right)\right\Vert _{\infty}\leq\bar{\xi}_{3}.
\end{equation}
\label{ass:1}
\end{assumption}

\begin{assumption}
The UAV torques are bounded, i.e.,
\begin{equation}
    \tau_{\phi}\leq\tau_{max}, \tau_{\theta}\leq\tau_{max}, \tau_{\psi}\leq\tau_{max}.
\end{equation}
\label{ass:2}
\end{assumption}
Those assumptions are reasonable for UAVs in practice due to their physical constraints and limits \cite{6842371}.
% \textcolor{red}{(cite toi cac bai bao co assumption tuong tu neu co the)}.

\begin{proposition}
Consider the UAV dynamics in (\ref{eqn:at_dynamic_2}), $\boldsymbol{G}_1$, $\boldsymbol{G}_2$ and $g\left(\boldsymbol{\xi}_2\right)$ are then bounded, i.e.,
\begin{equation}
    \left\Vert \boldsymbol{G}_{1}\right\Vert _{\infty}\leq\bar{I},\left\Vert \boldsymbol{G}_{2}\right\Vert _{\infty}\leq\bar{l}_{1},\left\Vert \boldsymbol{G}_{2}^{-1}\right\Vert _{\infty}\leq\bar{l}_{2},\left\Vert g\left(\boldsymbol{\xi}_{2}\right)\right\Vert _{\infty}\leq\bar{\xi}_{2}^{2}.
\end{equation}
\label{pro:1}
\end{proposition}

\begin{proof}
Since $\phi$, $\theta$, $\psi$ are the elements of $\boldsymbol{\xi}_1$, we have $\left\Vert \phi\right\Vert _{\infty}\leq\bar{\xi}_{1}$, $\left\Vert \theta\right\Vert _{\infty}\leq\bar{\xi}_{1}$, and $\left\Vert \psi\right\Vert _{\infty}\leq\bar{\xi}_{1}$. Similarly, $\left\Vert \dot{\phi}\right\Vert _{\infty}\leq\bar{\xi}_{2}$, $\left\Vert \dot{\theta}\right\Vert _{\infty}\leq\bar{\xi}_{2}$, and $\left\Vert \dot{\psi}\right\Vert _{\infty}\leq\bar{\xi}_{2}$. According to Lemma \ref{lem:1}, it follows that
\begin{equation}
    \left\Vert g\left(\boldsymbol{\xi}_{2}\right)\right\Vert _{\infty}=\left\Vert \left[\begin{array}{c}
    \dot{\theta}\dot{\psi}\\
    \dot{\phi}\dot{\psi}\\
    \dot{\phi}\dot{\theta}
    \end{array}\right]\right\Vert _{\infty}\leq\bar{\xi}_{2}^{2}<\infty.
\end{equation}

On the other hand, as the values of UAV parameters are bounded, we have
\begin{equation}
    \left\Vert \boldsymbol{G}_{1}\right\Vert _{\infty}=\max\left\{ \dfrac{I_{y}-I_{z}}{I_{x}},\dfrac{I_{z}-I_{x}}{I_{y}},\dfrac{I_{x}-I_{y}}{I_{z}}\right\} \leq\bar{I}<\infty,
\end{equation}
\begin{equation}
    \left\Vert \boldsymbol{G}_{2}\right\Vert _{\infty}=\max\left\{ \dfrac{l_{a}}{I_{x}},\dfrac{l_{a}}{I_{y}},\dfrac{1}{I_{z}}\right\} \leq\bar{l}_{1}<\infty,
\end{equation}
\begin{equation}
    \left\Vert \boldsymbol{G}_{2}^{-1}\right\Vert _{\infty}=\max\left\{ \dfrac{I_{x}}{l_{a}},\dfrac{I_{y}}{l_{a}},I_{z}\right\} \leq\bar{l}_{2}<\infty.
\end{equation}
\end{proof}

\begin{proposition}
Consider the UAV being controlled by the SMC technique (\ref{eqn:U} - \ref{eqn:U_detail}), the error signal $\boldsymbol{z}_2$ and the sliding surface $s$ are then bounded with $\left\Vert \lambda\right\Vert _{\infty}=\bar{\lambda}$.
\label{pro:2}
\end{proposition}

\begin{proof}
Since $\dot{V}_{SMC}\leq0$, it follows that
\begin{equation}
    \left\Vert \boldsymbol{Z}\left(t\right)\right\Vert _{\infty}\leq\left\Vert \boldsymbol{Z}\left(0\right)\right\Vert _{\infty}\leq\left\Vert \boldsymbol{Z}\left(0\right)\right\Vert _{2}.
\end{equation}

Besides, we have $\left\Vert \boldsymbol{z}_{1}\right\Vert _{\infty}\leq\left\Vert \boldsymbol{Z}\right\Vert _{\infty}$ and $\left\Vert \boldsymbol{z}_{2}\right\Vert _{\infty}\leq\left\Vert \boldsymbol{Z}\right\Vert _{\infty}$. Let $\left\Vert \boldsymbol{Z}\left(0\right)\right\Vert _{2}=\bar{Z}$, we have $\left\Vert \boldsymbol{z}_{1}\right\Vert _{\infty}\leq\bar{Z}$ and $\left\Vert \boldsymbol{z}_{2}\right\Vert _{\infty}\leq\bar{Z}$. As the result, the sliding surface is bounded, i.e., 
\begin{equation}
    \left\Vert s\right\Vert _{\infty}=\left\Vert \boldsymbol{z}_{2}+\lambda\boldsymbol{z}_{1}\right\Vert_{\infty} \leq\left(1+\bar{\lambda}\right)\bar{Z}.
\end{equation}
\end{proof}

\begin{theorem}
% \textbf{\textcolor{red}{Bo sung them chi tiet de ro rang hon cho phan chung minh Theorem 1 va 2 nhu hom truoc ta thao luan}}
Suppose the control gain matrices are chosen as $c_1$ and $c_2$. Let $\left\Vert c_{1}\right\Vert _{\infty}=\bar{c}_{1}$, $\left\Vert c_{2}\right\Vert _{\infty}=\bar{c}_{2}$. For $h\left(\boldsymbol{\xi}\right) = \boldsymbol{U}_\text{SMC}$, 
the LNMPC (\ref{eqn:lnmpc} - \ref{eqn:lnmpc1}) admits recursive feasibility for all $t\geq0$ if the following relation is satisfied:
\begin{equation}
    \bar{l}_{2}\left(\bar{\xi}_{3d}+\bar{\lambda}\bar{Z}+\bar{I}\bar{\xi}_{2}^{2}\right)+\bar{l}_{2}\left[\bar{c}_{1}+\bar{c}_{2}\left(1+\bar{\lambda}\right)\bar{Z}\right]\leq\tau_{max}.
    \label{the:1}
\end{equation}
\end{theorem}

\begin{proof}
Let $\boldsymbol{U}_{max}=\left[\tau_{\phi max},\tau_{\theta max},\tau_{\psi max}\right]^{T}$. Given the present system state $\boldsymbol{\xi}\left(t\right)$, $h\left(\boldsymbol{\xi}\right)$ is always feasible for the LNMPC problem (\ref{eqn:lnmpc} - \ref{eqn:lnmpc1}) if $\left\vert h\left(\boldsymbol{\xi}\right)\right\vert\leq\boldsymbol{U}_{max}$.

By taking the norm of both sides of (\ref{eqn:U}) and noting Proposition \ref{pro:1} and \ref{pro:2}, we get 
\begin{equation}
\begin{aligned}
    \left\Vert U_\text{SMC}\right\Vert _{\infty}&\leq\left\Vert \boldsymbol{U}_{eq}\right\Vert _{\infty}+\left\Vert \boldsymbol{U}_{sw}\right\Vert _{\infty}\\
    &=\left\Vert \boldsymbol{G}_{2}^{-1}\left(\dot{\boldsymbol{\xi}}_{2d}-\lambda\boldsymbol{z}_{2}-\boldsymbol{G}_{1}g\left(\boldsymbol{\xi}_{2}\right)\right)\right\Vert +\left\Vert -\boldsymbol{G}_{2}^{-1}\left(c_{1}\text{sign}\left(s\right)+c_{2}s\right)\right\Vert _{\infty}\\
    &\leq\bar{l}_{2}\left(\bar{\xi}_{3d}+\bar{\lambda}\bar{Z}+\bar{I}\bar{\xi}_{2}^{2}\right)+\bar{l}_{2}\left[\bar{c}_{1}+\bar{c}_{2}\left(1+\bar{\lambda}\right)\bar{Z}\right].
\end{aligned}
\label{eqn:the1}
\end{equation}
From (\ref{the:1}) and (\ref{eqn:the1}), we have $\left\Vert h\left(\boldsymbol{\xi}\right)\right\Vert _{\infty}=\left\Vert \boldsymbol{U_\text{SMC}}\right\Vert _{\infty}\leq\tau_{max}$, which completes the proof.
\end{proof}

Noting that the SMC acts as an initial guess for the optimization problem (\ref{eqn:lnmpc} - \ref{eqn:lnmpc1}) rather than being used in the optimization process. Therefore, the feasibility is guaranteed by the SMC controller, and the control performance is optimized by the LNMPC controller.

\begin{theorem}
The closed-loop system under LNMPC is asymptotically stable with respect to the equilibrium $\boldsymbol{Z}=\left[\begin{array}{c}
\boldsymbol{z}_{1}\\
\boldsymbol{z}_{2}
\end{array}\right]=\left[\begin{array}{c}
0\\
0
\end{array}\right]$ and the UAV attitude will converge to the desired reference under the LNMPC.
\end{theorem}

\begin{proof}
Since the Lyapunov function $V_\text{SMC}$ in (\ref{eqn:V}) is continuously differentiable and radially unbounded, by using the converse Lyapunov theorems \cite{Khalil}, there exist functions $\vartheta_i,i=1,2,3$ that satisfy the following inequalities
\begin{equation}
\begin{aligned}
    &\vartheta_{1}\left(\left\Vert \boldsymbol{\xi}\right\Vert \right)\leq V_\text{SMC}\left(\boldsymbol{\xi}\right)\leq\vartheta_{2}\left(\left\Vert \boldsymbol{\xi}\right\Vert \right),\\
    &\dfrac{\partial V\left(\boldsymbol{\xi}\right)}{\partial\boldsymbol{\xi}}f\left(\boldsymbol{\xi},h\left(\boldsymbol{\xi}\right)\right)\leq-\vartheta_{3}\left(\left\Vert \boldsymbol{\xi}\right\Vert \right).
\end{aligned}
\end{equation}

Consider contraction constraint (\ref{eqn:contraction}) and note that the optimal solution will be implemented for one sampling period each time, we have 
\begin{equation}
    \dfrac{\partial V\left(\boldsymbol{\xi}\right)}{\partial\boldsymbol{\xi}}f\left(\boldsymbol{\xi},\boldsymbol{U}\right)\leq\dfrac{\partial V\left(\boldsymbol{\xi}\right)}{\partial\boldsymbol{\xi}}f\left(\boldsymbol{\xi},h\left(\boldsymbol{\xi}\right)\right)\leq-\vartheta_{3}\left(\left\Vert \boldsymbol{\xi}\right\Vert \right).
\end{equation}

Based on Theorem 4.16 in \cite{Khalil}, we conclude that the closed-loop system is asymptotically stable with a guaranteed region of attraction:
\begin{equation}
    \mathbb{X}=\left\{\boldsymbol{\xi}\in\mathbb{R}^{6}\vert\bar{l}_{2}\left(\bar{\xi}_{3d}+\bar{\lambda}\bar{Z}+\bar{I}\bar{\xi}_{2}^{2}\right)+\bar{l}_{2}\left[\bar{c}_{1}+\bar{c}_{2}\left(1+\bar{\lambda}\right)\bar{Z}\right]\leq\tau_{max}\right\}.
\end{equation}

In addition, since there is no other constraints on control gains $\bar{\lambda}$, $\bar{c}_1$, $\bar{c}_2$, $\mathbb{X}$ can be enlarged by reducing the magnitude of those control gains as long as (\ref{the:1}) is satisfied. The LNMPC is thus recursively feasible, and the closed-loop system is stable.
\end{proof}

\begin{remark}
Condition (\ref{the:1}) can be satisfied by choice of $c_1$, $c_2$ and $\lambda$ as the remaining parameters are bounded and can be pre-determined. Specifically, UAV parameters represented in (\ref{eqn:at_dynamic_2}) are bounded according to Proposition \ref{pro:1} and their values can be determined based on the physical limits of the UAV. Similarly, the reference attitude values, control signals and state limits are all specified. Since $\dot{V}_{SMC}\leq0$ is satisfied, the system error is decreased over time leading to the initial error to be considered as the bounded error $\boldsymbol{Z}$. Therefore, (\ref{the:1}) can be fulfilled by varying the values of $\bar{c}_1$, $\bar{c}_2$, and $\bar{\lambda}$.
\end{remark}

\begin{remark}
In practice, the proposed controller can be implemented in the embedded computer system of UAVs for real-time control since the computation time to solve the LNMPC problem (\ref{eqn:lnmpc} - \ref{eqn:lnmpc1}) can be effectively handled by the optimized libraries such as ACADO \cite{Houska2011a}. In addition, software and libraries like CasADi \cite{Andersson2019}, FORCES Pro \cite{zanelli2020forces}, and VIATOC \cite{kalmari2015toolkit} can also be used to handle the LNMPC problem in different UAV platforms.
\end{remark}
\section{Results}
\label{results}
In this section, the performance of the proposed LNMPC controller is evaluated and compared with two state-of-the-art nonlinear controllers, the backstepping controller (BSC) and the sliding mode controller (SMC). The UAV model used is the AscTec Pelican as shown in Figure \ref{fig:vehicle} with parameters shown in Table \ref{tbl:params}.
\begin{table}[h!]
\centering
\caption{UAV parameters}
\label{tbl:params}
\begin{tabular}{|l|l|l|}
\hline
$m$         & 1.0       & Mass of the UAV [kg]                                    \\ \hline
$g$         & 9.8       & Gravity [m/s$^2$]                                             \\ \hline
$Ix$, $Iy$  & 0.01      & Moment of inertia about Bx, By axis, respectively [kg.m$^2$]  \\ \hline
$Iz$        & 0.02      & Moment of inertia about Bz axis [kg.m$^2$]                    \\ \hline
$l_a$       & 0.21      & UAV arm length [m]                                      \\ \hline
\end{tabular}
\end{table}

Constraints on the UAV are defined as
\begin{equation}
\begin{aligned}
    &-\dfrac{\pi}{2}\leq\phi,\theta\leq\dfrac{\pi}{2},\\
    &-\dfrac{\pi}{2}\leq\dot{\phi},\dot{\theta},\dot{\psi}\leq\dfrac{\pi}{2},\\
    &-0.1\leq\tau_{\phi},\tau_{\theta},\tau_{\psi}\leq0.1.
\end{aligned}
\label{eq:limits}
\end{equation}

For the LNMPC, the sampling period is chosen as $dt = 0.02$ s, the predictive horizon is chosen as $T=30dt$, and the weight matrices are chosen as
\begin{equation}
    \begin{aligned}
    &\boldsymbol{P}=\boldsymbol{Q}=\text{diag}\left(30.0,30.0,30.0,1.0,1.0,1.0\right),\\
    &\boldsymbol{R}=\text{diag}\left(1.0,1.0,1.0\right).
    \end{aligned}
\end{equation}

% \textcolor{red}{Kiem tra lai doan nay xem co van de gi khong }
Since the BSC and SMC do not consider the UAV constraints, their control signals are limited to $u \in \{u_{max},u_{min}\}$ as follow:
\begin{equation}
    u=\min\left\{ u_{max},\max\left\{ u,u_{min}\right\} \right\}.
\end{equation}
Comparisons are then conducted in four scenarios with different input references.

\subsection{Scenario 1: Step input}
This scenario considers a step input. The UAV is initialized at a hovering state where all attitude angles are zeros. The new reference angles are then set to $\phi_d=\theta_d=\psi_d=1$ rad. The system responses are shown in Figure \ref{fig:tracking_step}. It can be seen that all controllers can drive the UAV to reach the reference without overshooting. However, the proposed controller smoothly drives the angles to their reference values within 1 s, while other controllers require a longer time, i.e., 2 s with BSC and 2.5 s with SMC. The proposed controller thus introduces the fastest response due to its ability to predict future states to generate more efficient control signals.

\begin{figure}[h!]
\centering
    \centering
    \includegraphics[width=0.6\textwidth]{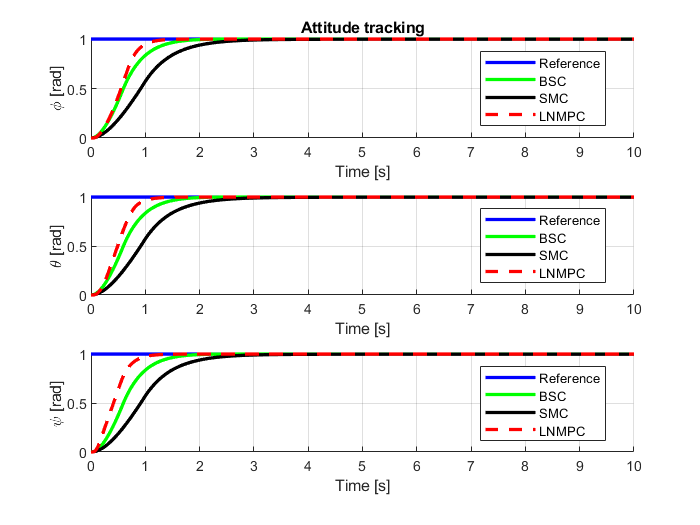}
    \caption{The step responds of three controllers in Scenario 1}
    \label{fig:tracking_step}
\end{figure}

% \begin{figure}[h!]
% \centering
%     \begin{subfigure}[b]{0.4\textwidth}
%     \centering
%     \includegraphics[width=\textwidth]{images/control_bsc_step.png}
%     \caption{}
%     \label{fig:bsc_step}
%     \end{subfigure}
%     \begin{subfigure}[b]{0.4\textwidth}
%     \centering
%     \includegraphics[width=\textwidth]{images/control_smc_step.png}
%     \caption{}
%     \label{fig:smc_step}
%     \end{subfigure}
%     \begin{subfigure}[b]{0.4\textwidth}
%     \centering
%     \includegraphics[width=\textwidth]{images/control_lnmpc_step.png}
%     \caption{}
%     \label{fig:lnmpc_step}
%     \end{subfigure}
%     \caption{The control signals of three controllers (Scen.1). \ref{fig:bsc_step}: Backstepping Control (BSC). \ref{fig:smc_step}: Sliding Mode Control (BSC). \ref{fig:lnmpc_step}: Proposed Lyapunov-based Nonlinear Model Predictive Control (LNMPC)}
%     \label{fig:control_step}
% \end{figure}

\subsection{Scenario 2: Time-varying reference}
In this scenario, the desired attitudes periodically change over time according to the following equations: 
\begin{equation}
\begin{aligned}
    &\phi_{d}=\dfrac{1}{3}\sin\left(t\right)+\dfrac{1}{2}\cos\left(2t\right)\\
    &\theta_{d}=\dfrac{1}{3}\cos\left(t\right)+\dfrac{1}{2}\sin\left(2t\right)\\
    &\psi_{d}=\dfrac{\pi}{10}t+\dfrac{1}{4}\cos\left(2t\right)
\end{aligned}
\end{equation}

\begin{figure}[h!]
\centering
    \begin{subfigure}[b]{0.49\textwidth}
    \centering
    \includegraphics[width=\textwidth]{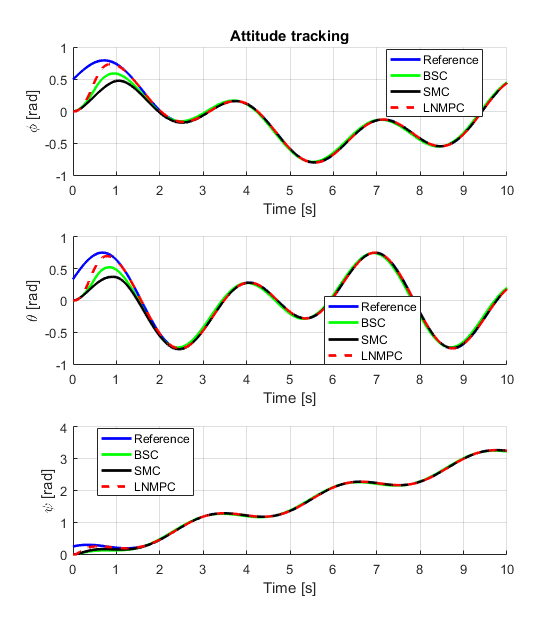}
    \caption{Tracked trajectories}
    \label{fig:trajectories}
    \end{subfigure}
    \begin{subfigure}[b]{0.49\textwidth}
    \centering
    \includegraphics[width=\textwidth]{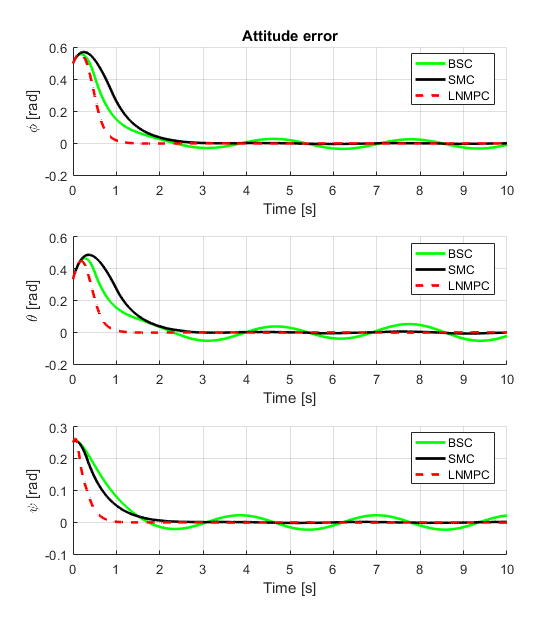}
    \caption{Tracking errors}
    \label{fig:error}
    \end{subfigure}
    \caption{Trajectory tracking results of three controllers for Scenario 2}
    \label{fig:tracking}
\end{figure}

\begin{figure}[h!]
\centering
    \begin{subfigure}[b]{0.4\textwidth}
    \centering
    \includegraphics[width=\textwidth]{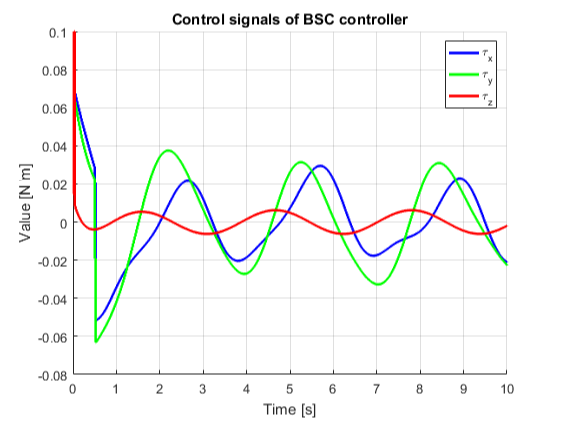}
    \caption{BSC}
    \label{fig:bsc}
    \end{subfigure}
    \begin{subfigure}[b]{0.4\textwidth}
    \centering
    \includegraphics[width=\textwidth]{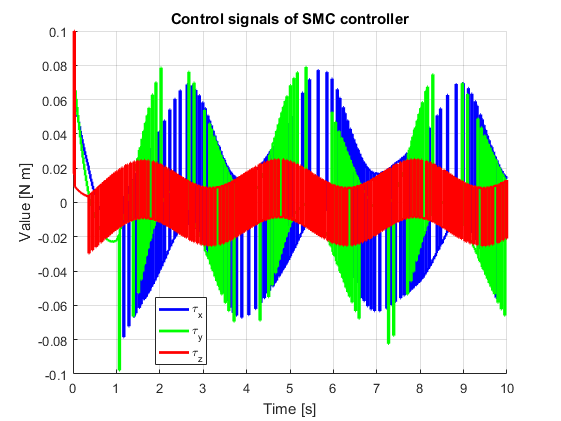}
    \caption{SMC}
    \label{fig:smc}
    \end{subfigure}
    \begin{subfigure}[b]{0.4\textwidth}
    \centering
    \includegraphics[width=\textwidth]{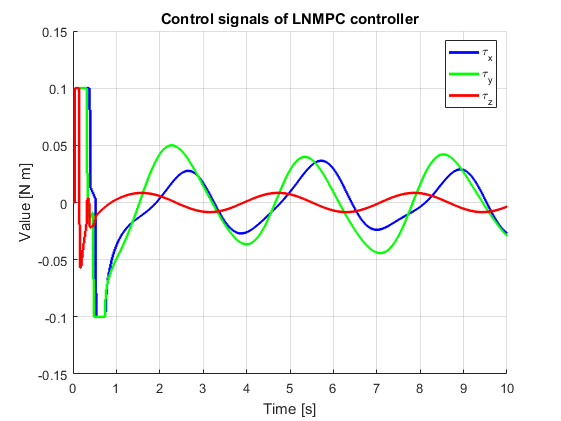}
    \caption{LNMPC}
    \label{fig:lnmpc}
    \end{subfigure}
    \caption{Control signals of three controllers in Scenario 2}
    \label{fig:control}
\end{figure}

Figure \ref{fig:tracking} shows the tracking results. It can be seen that the LNMPC introduces the fastest transient response by first reaching the references. At the steady-state, the LNMPC also has the smallest tracking errors. The BSC, on the other hand, introduces a relatively fast response. However, as shown in Figure \ref{fig:error}, its performance is affected by the system constraints causing large steady-state errors. The SMC introduces comparable tracking performance to the LNMPC as it is not constrained after reaching the desired reference. However, it has the slowest response and generates chattering control signals as shown in Figure \ref{fig:control}.

\subsection{Scenario 3: Saturated control signal}
To further evaluate the performance of the proposed controller, more challenging reference trajectories are used in Scenario 3, in which control signals become saturated in a certain period. The equations of those trajectories are given by: 
\begin{equation}
\begin{aligned}
    &\phi_{d}=\dfrac{1}{3}\sin\left(2t\right)+\dfrac{1}{6}\sin\left(6t\right)\\
    &\theta_{d}=\dfrac{1}{3}\cos\left(2t\right)+\dfrac{1}{6}\cos\left(4t\right)\\
    &\psi_{d}=\dfrac{\pi}{10}t+\dfrac{1}{5}\sin\left(3t\right)+0.5
\end{aligned}
\end{equation}

\begin{figure}[h!]
\centering
    \begin{subfigure}[b]{0.49\textwidth}
    \centering
    \includegraphics[width=\textwidth]{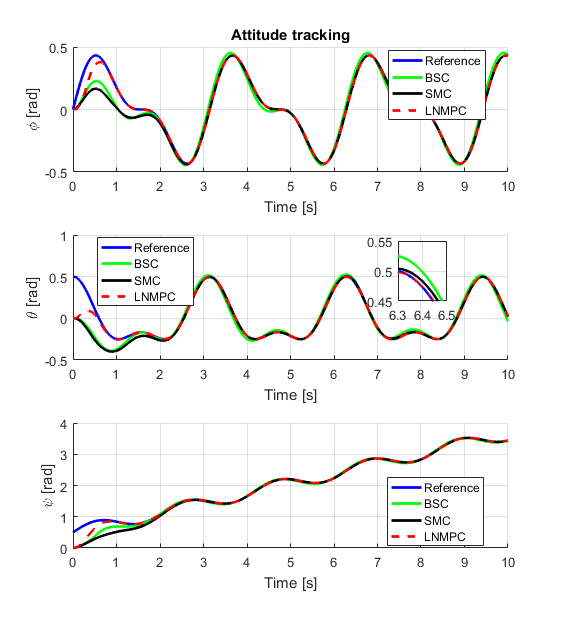}
    \caption{Tracked trajectories}
    \label{fig:trajectories2}
    \end{subfigure}
    \begin{subfigure}[b]{0.49\textwidth}
    \centering
    \includegraphics[width=\textwidth]{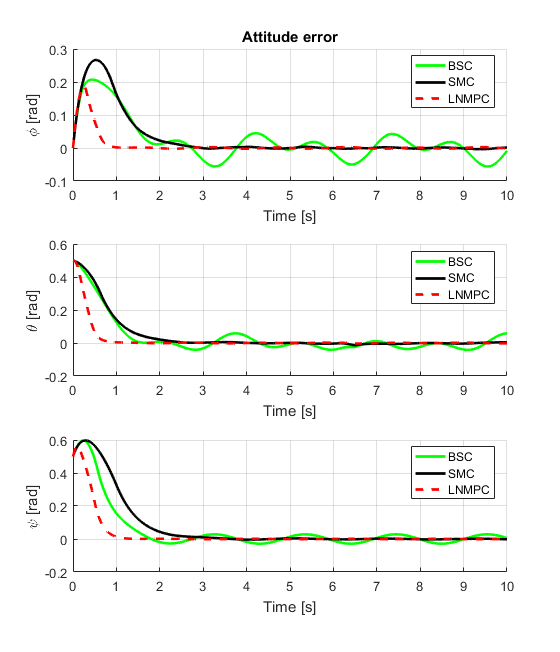}
    \caption{Tracking errors}
    \label{fig:error2}
    \end{subfigure}
    \caption{Tracking trajectories and errors of three controllers in Scenario 3}
    \label{fig:tracking2}
\end{figure}

\begin{figure}[h!]
\centering
    \begin{subfigure}[b]{0.4\textwidth}
    \centering
    \includegraphics[width=\textwidth]{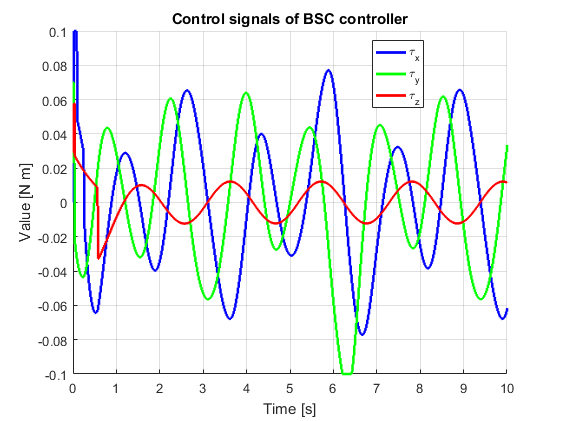}
    \caption{BSC}
    \label{fig:bsc2}
    \end{subfigure}
    \begin{subfigure}[b]{0.4\textwidth}
    \centering
    \includegraphics[width=\textwidth]{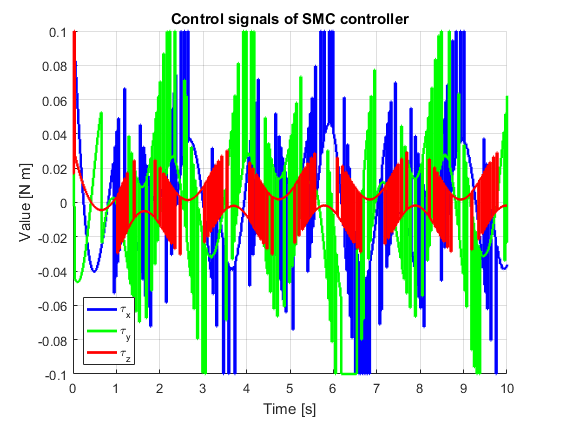}
    \caption{SMC}
    \label{fig:smc2}
    \end{subfigure}
    \begin{subfigure}[b]{0.4\textwidth}
    \centering
    \includegraphics[width=\textwidth]{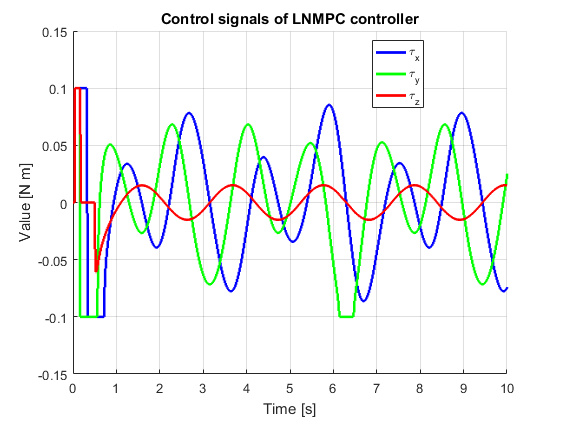}
    \caption{LNMPC}
    \label{fig:lnmpc2}
    \end{subfigure}
    \caption{Control signals of three controllers in Scenario 3}
    \label{fig:control2}
\end{figure}

Figure \ref{fig:tracking2} shows the tracking results. Similar to previous scenarios, the LNMPC introduces the fastest transient responses and smallest steady-state errors. Its control signals also do not exhibit the chattering phenomenon as with the SMC, as shown in Figure \ref{fig:control2}. Especially, the advantage of the LNMPC is clearly shown in period $6<t<7$ s, where the control signals from all controllers are saturated due to the system constraints. The tracking errors of the BSC and SMC for pitch angle $\theta$ quickly increase, whereas that error of the LNMPC remains the same, as can be seen in Figure \ref{fig:trajectories2}. That advantage comes from the capability of the LNMPC to predict its future states and then provide optimal control signals to better adapt to their changes.

\subsection{Scenario 4: Disturbance rejection}
This scenario evaluates the robustness of our controller in the presence of input and output disturbances. The reference signals in this scenario are given by:

\begin{center}
\begin{equation}
\begin{aligned}
    &\phi_{d}=\dfrac{1}{3}\cos\left(2t\right)+\dfrac{1}{2}\cos\left(\dfrac{2}{3}t\right)\\
    &\theta_{d}=\dfrac{1}{2}\cos\left(2t\right)+\dfrac{1}{6}\sin\left(\dfrac{2}{3}t\right)\\
    &\psi_{d}=\dfrac{\pi}{12}t+\cos\left(\dfrac{3}{5}t\right)
\end{aligned}
\end{equation}
\end{center}

\begin{figure}[h!]
\centering
    \begin{subfigure}[b]{0.49\textwidth}
    \centering
    \includegraphics[width=\textwidth]{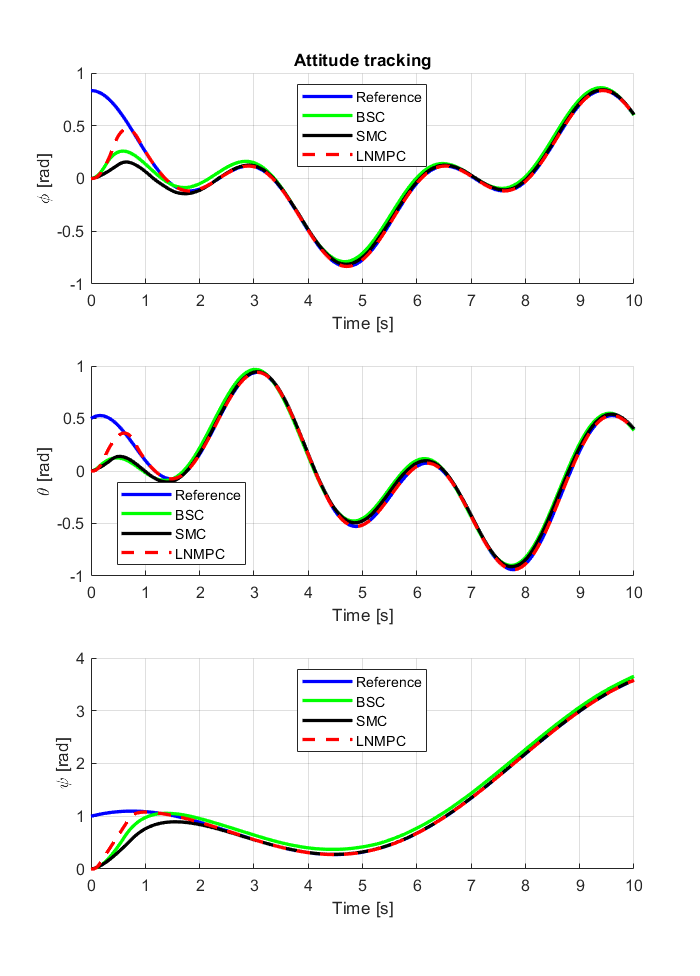}
    \caption{Tracked trajectories}
    \label{fig:trajectories3}
    \end{subfigure}
    \begin{subfigure}[b]{0.49\textwidth}
    \centering
    \includegraphics[width=\textwidth]{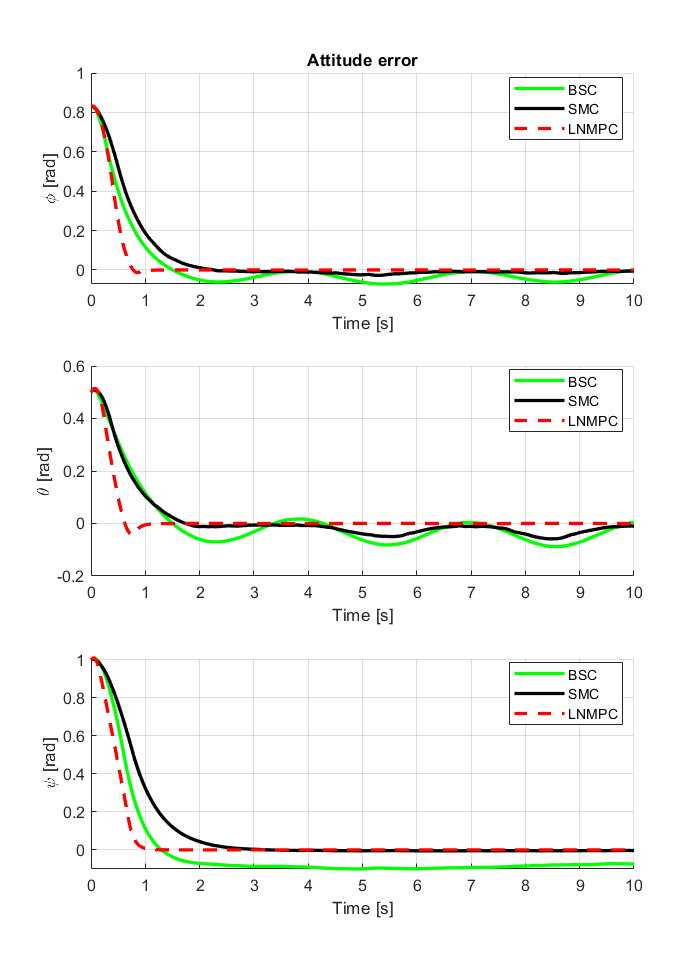}
    \caption{Tracking errors}
    \label{fig:error3}
    \end{subfigure}
    \caption{Tracking trajectories and errors of three controllers in the presence of input disturbances}
    \label{fig:tracking4}
\end{figure}

\begin{figure}[h!]
\centering
    \begin{subfigure}[b]{0.49\textwidth}
    \centering
    \includegraphics[width=\textwidth]{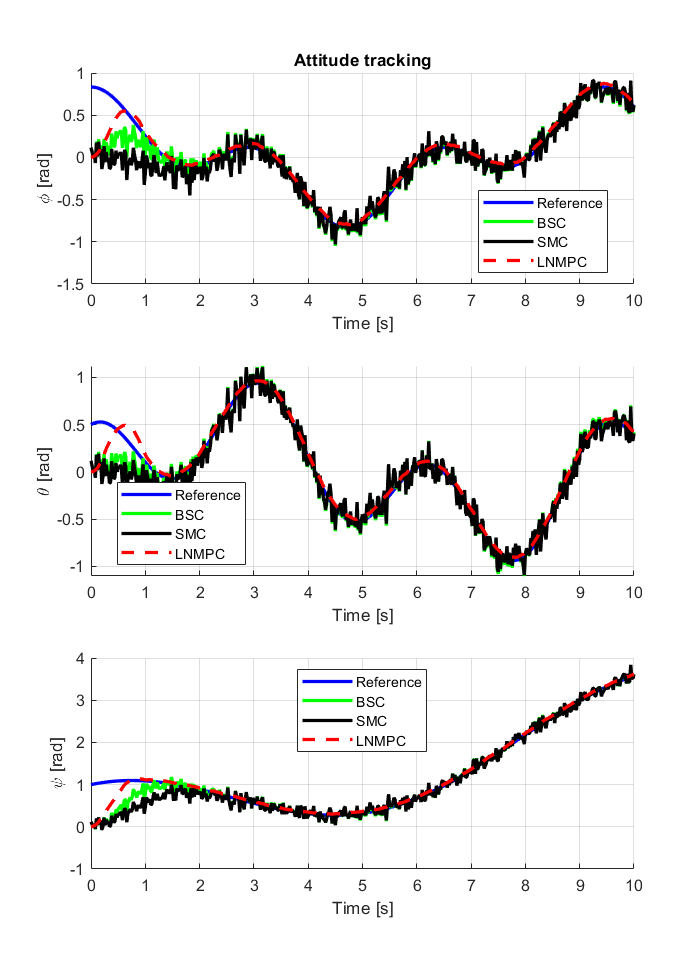}
    \caption{Tracked trajectories}
    \label{fig:trajectories4_1}
    \end{subfigure}
    \begin{subfigure}[b]{0.49\textwidth}
    \centering
    \includegraphics[width=\textwidth]{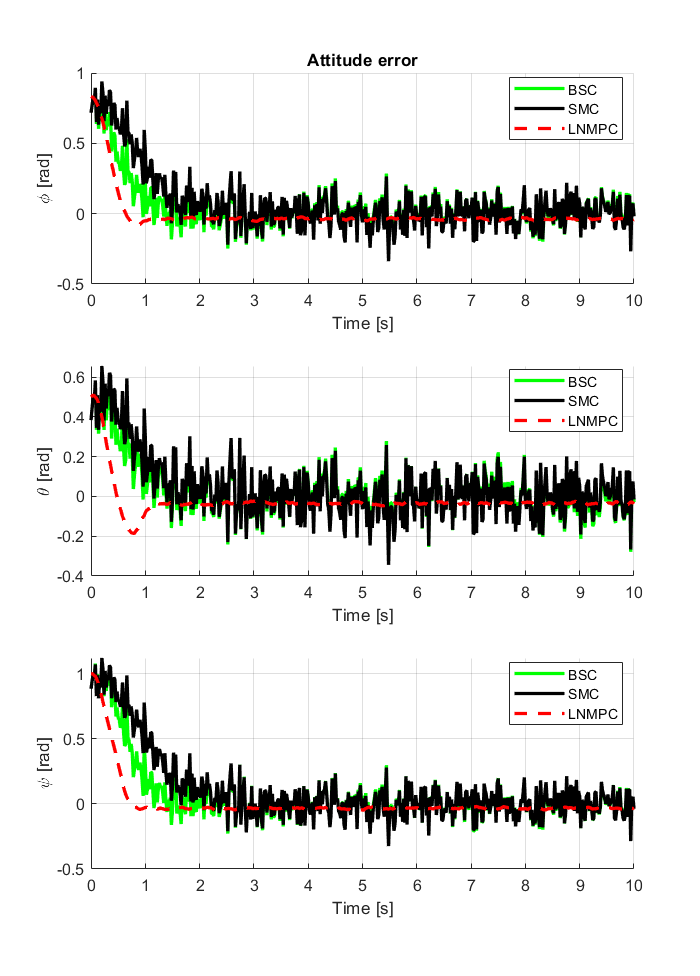}
    \caption{Tracking errors}
    \label{fig:error4_1}
    \end{subfigure}
    \caption{Tracking trajectories and errors of three controllers in the presence of output disturbances}
    \label{fig:tracking4_1}
\end{figure}
The input disturbance is generated from uniformly distributed random noise with magnitude $d_u = 0.1$ Nm added to the control signals. Note that this magnitude is equal to the maximum input torque mentioned in ({\ref{eq:limits}}). The tracking results in Figure {\ref{fig:tracking4}} show that all controllers are rather insensitive to input disturbances. However, the LNMPC performs the best in terms of transient response and steady state error.

For the case of output disturbance, the system states are added with Gaussian noise of zero mean and variance equal to 0.02 rad. The tracking results are depicted in Figure {\ref{fig:tracking4_1}}. It can be seen that the trajectories of BSC and SMC fluctuate about their references implying the influence of noise on those controllers. The LNMPC, on the other hand, copes well with the noise to maintain stable tracking performance. The LNMPC therefore is sufficiently robust to deal with disturbances appearing during UAV operation such as measurement noise and aerodynamic disturbance.

\subsection{Validation with software-in-the-loop tests based on Gazebo simulator}
\begin{figure}[h!]
    \centering
    \includegraphics[width=0.6\textwidth]{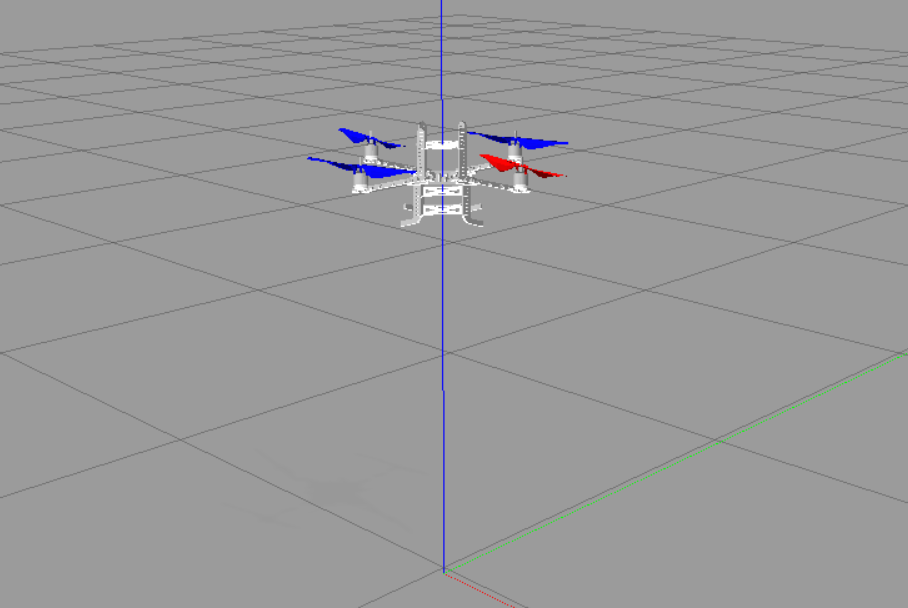}
    \caption{The quadrotor UAV model in the Gazebo simulator}
    \label{fig:gazebo}
\end{figure}

To validate the control performance of the LNMPC on flight tasks, we carry out software-in-the-loop (SIL) tests in which the proposed controller is implemented to control a UAV model created in the Gazebo simulating environment as shown in Figure \ref{fig:gazebo}. Parameters of the UAV are inherited from the AscTec Pelican drone, as shown in Figure \ref{fig:vehicle}. Since the LNMPC only controls the attitude, we use a standard nonlinear MPC to handle the altitude and position tracking. The results are then compared with the built-in PID controller of the UAV. 

The SIL simulation includes two scenarios. One is circular, and the other is a sinusoidal trajectory. The trajectory tracking results are shown in Figure \ref{fig:rviz} where the UAV can reach and track the reference trajectories. The attitude tracking results are shown in Figure \ref{fig:sitl} where it can be seen that the LNMPC outperforms the PID controller in both the response speed and tracking error. This result can be further verified by Table \ref{tbl:rmse} which shows the root-mean-square error (RMSE). The LNMPC introduces smaller errors in all roll, pitch, and yaw angles, implying its better performance compared to the PID controller.

To evaluate the computational practicality of our proposed controller, we have implemented it on an embedded computer named Jetson Nano, a popular embedded
computer used for UAVs \cite{9416140,9359491}. Our implementation is based on the ACADO library {\cite{Houska2011a}}. The average calculation time for the LNMPC in each processing cycle is 0.016834 s, less than the sampling time, 0.02 s. The result thus confirms the validity of our proposed controller for real UAVs.

\begin{figure}[h!]
\centering
    \begin{subfigure}[b]{0.49\textwidth}
    \centering
    \includegraphics[width=\textwidth]{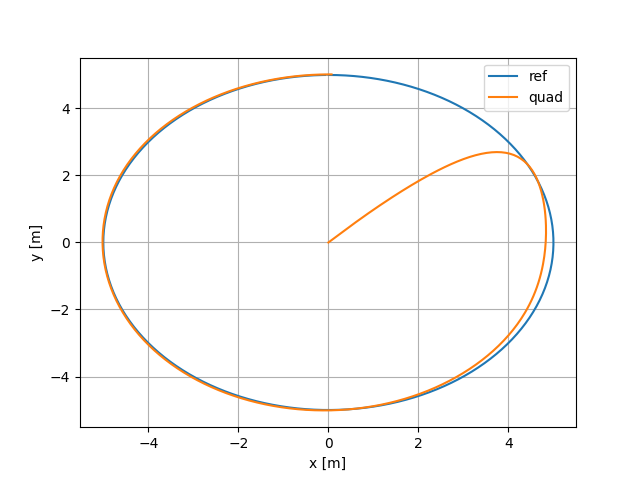}
    \caption{Scenario 1}
    \label{fig:topview1}
    \end{subfigure}
    \begin{subfigure}[b]{0.49\textwidth}
    \centering
    \includegraphics[width=\textwidth]{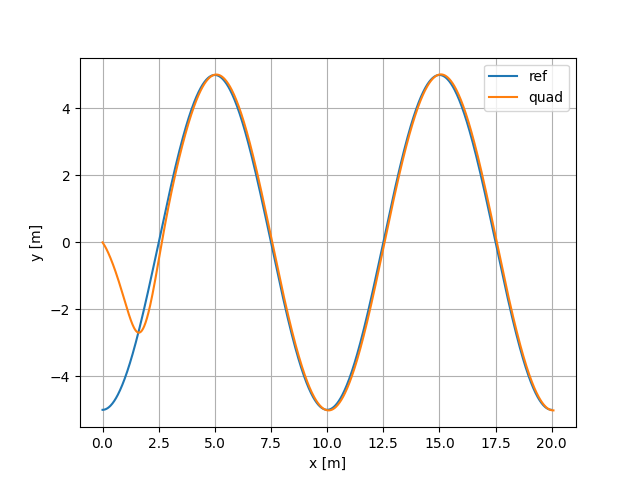}
    \caption{Scenario 2}
    \label{fig:topview2}
    \end{subfigure}
    \caption{Trajectory tracking results of the UAV (Top view)}
    \label{fig:rviz}
\end{figure}

\begin{figure}[h!]
\centering
    \begin{subfigure}[b]{0.49\textwidth}
    \centering
    \includegraphics[width=\textwidth]{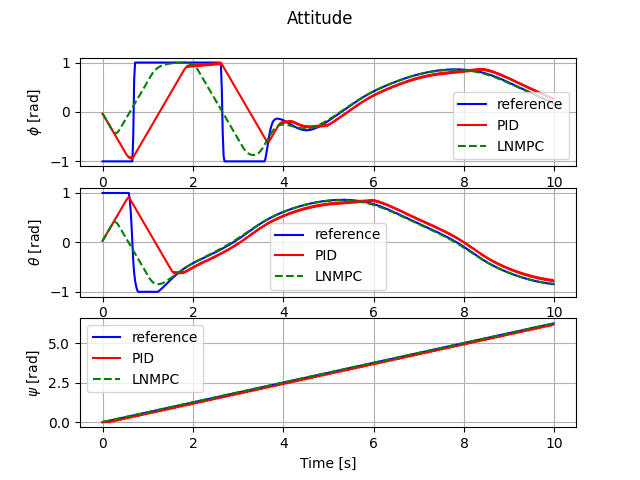}
    \caption{Scenario 1}
    \label{fig:val1}
    \end{subfigure}
    \begin{subfigure}[b]{0.49\textwidth}
    \centering
    \includegraphics[width=\textwidth]{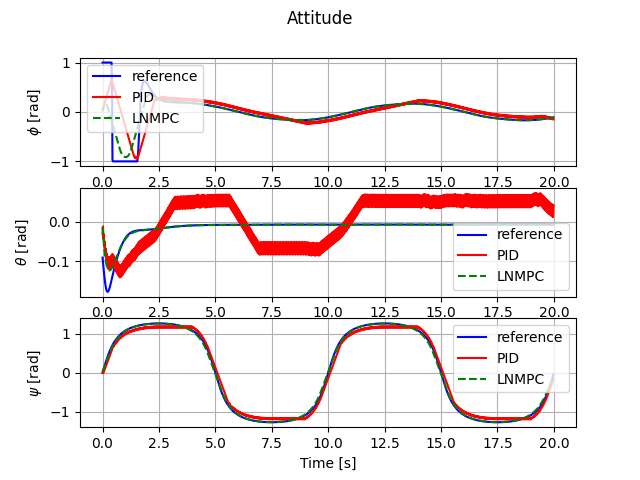}
    \caption{Scenario 2}
    \label{fig:val2}
    \end{subfigure}
    \caption{Attitude tracking results}
    \label{fig:sitl}
\end{figure}

% \begin{table}[h!]
% \centering
% \caption{The attitude RMSE of the controller (rad)}
% \label{tbl:rmse}
% \begin{tabular}{ p{1.0cm}||p{1.0cm} p{1.0cm} p{1.0cm} | p{1.0cm} p{1.0cm} p{1.0cm}} 
% % \hline
%       & $e_\phi$    & $e_\theta$    & $e_\psi$  & $e_\phi$    & $e_\theta$    & $e_\psi$  \\ 
% \hline
% PID   & 0.5661 | 0.0823     & 0.3650 | 0.2704       & 0.0842 | 0.7044  \\ 
% \hline
% LNMPC & 0.3139 | 0.0638     & 0.2248 | 0.1956       & 0.0120 | 0.3531  \\
% % \hline
% \end{tabular}
% \end{table}

\begin{table}[h!]
\centering
\caption{The attitude RMSE of the controller (rad)}
\label{tbl:rmse}
\begin{tabular}{ p{1.0cm}||p{1.0cm} p{1.0cm} p{1.0cm} | p{1.0cm} p{1.0cm} p{1.0cm}}
& \multicolumn{3}{c|}{Scenario 1}  & \multicolumn{3}{c}{Scenario 2}  \\ \cline{2-7} 
                         & $e_\phi$ & $e_\theta$ & $e_\psi$ & $e_\phi$ & $e_\theta$ & $e_\psi$ \\ \hline
PID                      & 0.5661   & 0.3650     & 0.0842  & 0.2848   & 0.0571     & 0.0937  \\ \hline
LNMPC                    & \textbf{0.3139}   & \textbf{0.2248}     & \textbf{0.0120}  & \textbf{0.1793}   & \textbf{0.0087}     & \textbf{0.0294}  \\ 
\end{tabular}
\end{table}

\section{Conclusion}
\label{SecConclusion}
% \textcolor{red}{Nam viet lai conclusion theo cau truc sau nhe:\\
% 1 cau De xuat da thuc hien trong bai bao\\
% 2 - 3 cau trinh bay tom gon nhung uu diem cua phuong phap de xuat\\
% 2 - 3 cau trinh bay nhung uu diem ve ket qua dat duoc va y nghia\\
% 1 cau trinh bay dinh huong future work (neu co)}

In this work, we have introduced a new LNMPC designed to track UAVs' attitude trajectories. By using a contraction condition and the design of a SMC, we proved the feasibility and closed-loop stability of the proposed controller. The results show that our controller outperforms other popular nonlinear controllers in response time, tracking accuracy, and control signal smoothness, especially in challenging situations where significant changes in control effort are required. In addition, the results of SIL tests with the Gazebo simulator show the ability to deploy and apply the proposed controller in the embedded control system of UAVs, which then can be used to replace their built-in PID controller. Our future work will extend the LNMPC for the formation control of multiple UAVs.

\section*{Conflict of interest statement}
On behalf of all authors, the corresponding author states that there is no conflict of interest.

\bibliography{sn-article}% common bib file


\end{document}